\documentclass[11pt]{article}
\usepackage[a4paper,margin=1in]{geometry}
\usepackage{amsmath,amssymb,amsthm,mathtools,bm}
\usepackage{booktabs,array}
\usepackage{graphicx}
\usepackage{microtype}
\usepackage{enumitem}
\usepackage[numbers,sort&compress]{natbib}
\usepackage[colorlinks=true,linkcolor=blue,citecolor=blue,urlcolor=blue]{hyperref}
\usepackage{cleveref}
\allowdisplaybreaks

\newtheorem{theorem}{Theorem}[section]
\newtheorem{proposition}[theorem]{Proposition}

\newtheorem{corollary}[theorem]{Corollary}
\newtheorem{remark}[theorem]{Remark}
\newtheorem{definition}[theorem]{Definition}

\newcommand{\cL}{\mathcal L}
\newcommand{\grad}{\operatorname{grad}}

\title{\bfseries Poisson Pencils, Lie Symmetries and Hamiltonian Reductions of a Coupled Nonlinear Wave System:\\[2mm]
Tangent KdV Geometry and Elliptic Moduli}
\author{Alvaro H. Salas\\
\small Department of Mathematics and Statistics, Universidad Nacional de Colombia, Sede Manizales\\
\small FIZMAKO Research Group, Manizales, Colombia\\
\small ORCID: 0000-0001-9343-6062}
\date{}

\begin{document}
\maketitle

\begin{abstract}
We study the two-field nonlinear dispersive system
\[
 u_t=u_{xxx}+6uu_x,\qquad v_t=v_{xxx}+6(uv)_x,
\]
viewed simultaneously as a coupled wave equation, as the tangent covering of Korteweg--de Vries (KdV), and as a Hamiltonian flow on a tangent Poisson manifold.  The main purpose is to make these viewpoints interact at theorem level.  First, we prove that the Magri Poisson pencil of KdV admits a complete tangent lift to an explicit compatible pair of matrix Hamiltonian operators.  The corresponding recursion operator has triangular tangent form and generates the lifted KdV hierarchy.  Second, we identify a five-dimensional point-symmetry algebra together with the infinite hierarchy of tangent generalized symmetries.  Third, reduction by the traveling-wave subgroup produces a four-dimensional Hamiltonian system that is the complete tangent lift of the scalar KdV profile dynamics.  On the nonsingular elliptic locus, the base profile is written in Weierstrass form and every tangent traveling wave is classified explicitly by derivatives with respect to the energy and integration constants.  A Wronskian identity gives a canonical tangent basis and yields a sharp monodromy criterion for periodic tangent waves.  Finally, the Galilean and scaling symmetries act explicitly on the reduced parameters: the Weierstrass invariants are Galilean invariants, have weights four and six under scaling, and the elliptic modular invariant $j$ is preserved by both actions.  Thus the Lie-symmetry orbit space of nonsingular traveling waves carries a natural map to the moduli space of elliptic curves.  The results provide a concrete bridge between Poisson pencils, tangent coverings, Lie symmetry reduction, Hamiltonian geometry and elliptic moduli for an integrable coupled nonlinear wave system.
\end{abstract}



\section{Introduction}

The Korteweg--de Vries equation is one of the canonical models in the modern theory of integrable systems.  Its Lax formulation, inverse-scattering theory, infinite conservation laws and bi-Hamiltonian structure form a prototype for the interaction between nonlinear partial differential equations and geometry; see, among many classical sources, \cite{GardnerGreeneKruskalMiura1967,Lax1968,Magri1978,Dorfman1993}.  At the same time, tangent and cotangent coverings have become standard geometric devices for encoding linearized equations, generalized symmetries and Hamiltonian operators in the geometry of differential equations \cite{KerstenKrasilshchikVerbovetsky2004,KrasilshchikVerbovetsky2011}.  In finite-dimensional Poisson geometry, complete lifts transport Poisson tensors from a manifold to its tangent bundle and preserve the Schouten bracket \cite{GrabowskiUrbanski1997,MitricVaisman2003,Vaisman1994}.  The purpose of the present paper is to put these observations together in an explicit nonlinear-wave setting and then push the resulting geometry through Lie reduction to elliptic traveling-wave moduli.

We consider the coupled system
\begin{equation}
\label{eq:tkdv}
 u_t=u_{xxx}+6uu_x,\qquad
 v_t=v_{xxx}+6(uv)_x.
\end{equation}
The second equation is exactly the Fr\'echet linearization of the first one along the direction $v$; hence \eqref{eq:tkdv} is the tangent covering of KdV.  Although its first component closes on $u$, the second component is genuinely coupled to the nonlinear base field and carries nontrivial tangent dynamics.  This makes \eqref{eq:tkdv} a useful test case for the compatibility of three structures that are often treated separately: a Poisson pencil on a function space, the Lie algebra of symmetry reductions, and the algebraic geometry of exact traveling waves.

The first structural result is an explicit complete lift of the Magri pencil.  With
\begin{equation}
\label{eq:scalarJ}
 J_0=D_x,\qquad J_1=D_x^3+4uD_x+2u_x,
\end{equation}
we obtain the matrix operators
\begin{equation}
\label{eq:liftedJ}
 \widehat J_0=
 \begin{pmatrix}0&D_x\\ D_x&0\end{pmatrix},\qquad
 \widehat J_1=
 \begin{pmatrix}
 0&J_1\\
 J_1&4vD_x+2v_x
 \end{pmatrix}.
\end{equation}
We show that every pencil $\widehat J_1-\lambda\widehat J_0$ is Hamiltonian, and that \eqref{eq:tkdv} is bi-Hamiltonian with respect to the complete lifts of two consecutive KdV Hamiltonians.  The corresponding recursion operator is
\begin{equation}
\label{eq:RhatIntro}
 \widehat R=
 \begin{pmatrix}
 R&0\\
 R'[v]&R
 \end{pmatrix},\qquad
 R=D_x^2+4u+2u_xD_x^{-1},\qquad
 R'[v]=4v+2v_xD_x^{-1}.
\end{equation}
Thus the tangent hierarchy is not an auxiliary construction but the natural complete lift of the KdV Lenard--Magri chain.

The second part concerns Lie symmetry.  We exhibit the point-symmetry algebra generated by translations, Galilean boosts, KdV scaling, and independent fiber scaling.  The tangent functor also lifts every generalized KdV symmetry $G[u]$ to the pair $(G,\ell_G(v))$, and this lift preserves commutators.  In this way finite-dimensional Lie reduction and the infinite integrable hierarchy coexist in one framework.

The most geometric part of the paper begins with traveling waves.  Reduction by the one-parameter subgroup generated by $\partial_t-c\partial_x$ gives
\begin{equation}
\label{eq:travelingIntro}
 U''+3U^2-cU=A,\qquad
 V''+(6U-c)V=B.
\end{equation}
The scalar profile equation is a one-degree-of-freedom Hamiltonian system, while the coupled four-dimensional reduction is its complete tangent Hamiltonian lift, with the constant $B$ acting as an affine unfolding in the tangent direction.  On the nonsingular elliptic locus the base solution is
\begin{equation}
\label{eq:wpIntro}
 U(z)=\frac c6-2\wp(z-z_0;g_2,g_3),
\end{equation}
where
\begin{equation}
\label{eq:ginvariantsIntro}
 g_2=\frac{c^2}{12}+A,
 \qquad
 g_3=-\frac E2-\frac{Ac}{12}-\frac{c^3}{216}.
\end{equation}
Our main reduced-dynamics theorem states that every tangent profile is generated by derivatives of $U$ with respect to the moduli:
\begin{equation}
\label{eq:VmoduliIntro}
 V=B\,U_A+\alpha U_z+\beta U_E.
\end{equation}
Moreover,
\begin{equation}
\label{eq:wronskianIntro}
 W(U_z,U_E)=1,
\end{equation}
so $(U_z,U_E)$ is a canonical fundamental pair for the homogeneous tangent equation.  For a periodic base orbit of period $T(A,E,c)$, the derivatives $U_A$ and $U_E$ acquire exactly computable phase drifts.  Consequently \eqref{eq:VmoduliIntro} is periodic if and only if
\begin{equation}
\label{eq:periodicCriterionIntro}
 B T_A+\beta T_E=0.
\end{equation}
This gives a closed monodromy classification without solving a second-order variational equation by quadrature.

Finally, we calculate the induced action of the Galilean and scaling symmetries on $(c,A,E)$.  Galilean boosts preserve $g_2$ and $g_3$ exactly, while scaling gives weights $4$ and $6$.  Hence
\begin{equation}
 j=1728\frac{g_2^3}{g_2^3-27g_3^2}
\end{equation}
is a Lie-symmetry invariant on the nonsingular elliptic locus.  This identifies the elliptic $j$-invariant as a natural coordinate on the complex orbit space of traveling-wave curves modulo the elementary KdV symmetries.

The novelty of the paper is therefore not a new KdV-type equation, nor a new isolated elliptic solution.  Instead, it is the explicit synthesis of complete Poisson lifting, tangent recursion, Lie reduction, moduli derivatives and elliptic monodromy for the coupled tangent system.  This distinction is important because complete lifts and the KdV hierarchy are classical individually, while the complete orbit classification by the elliptic $j$-invariant, the affine moduli-tangent decomposition, and their integration with the lifted Poisson/Lax structures are the points where the structures are made to interact.

\section{Variational setting and complete tangent lifts}

We work either on the periodic loop space or on rapidly decaying fields, so integration by parts produces no boundary terms.  Let $D_x$ denote total differentiation.  For a local functional
\[
 H[u]=\int h(u,u_x,\ldots,u_{(N)})\,dx,
\]
its variational derivative is denoted $\delta H/\delta u$.  A formally skew-adjoint differential operator $J$ is Hamiltonian when the induced bracket
\[
 \{F,G\}_J=\int \frac{\delta F}{\delta u}\,
 J\left(\frac{\delta G}{\delta u}\right)\,dx
\]
satisfies the Jacobi identity.  Compatibility of $J_0,J_1$ means that $J_1-\lambda J_0$ is Hamiltonian for every $\lambda$.

For a differential operator $P(u)$, let $P'[v]$ be its Fr\'echet derivative in the direction $v$.  We use the following formal variational analogue of the complete tangent lift of a Poisson tensor.

\begin{definition}[Complete tangent lift]
\label{def:completeLift}
For a scalar differential operator $P(u)$ define
\begin{equation}
\label{eq:completeLift}
 P^C(u,v)=
 \begin{pmatrix}
 0&P(u)\\
 P(u)&P'[v]
 \end{pmatrix}.
\end{equation}
For a local functional $H[u]$, define its complete lift by
\begin{equation}
\label{eq:functionalLift}
 H^C[u,v]=\left.\frac{d}{d\varepsilon}\right|_{\varepsilon=0}
 H[u+\varepsilon v]
 =\int \frac{\delta H}{\delta u}\,v\,dx.
\end{equation}
\end{definition}

The block pattern in \eqref{eq:completeLift} is the infinite-dimensional counterpart of the standard complete lift on tangent bundles.  The lower-right block records the variation of the base Poisson tensor along the tangent vector.

\begin{theorem}[Schouten compatibility of the complete lift]
\label{thm:schoutenLift}
Let $P$ and $Q$ be variational bivectors for which the formal Schouten brackets are defined.  Then
\begin{equation}
\label{eq:schoutenIdentity}
 [P^C,Q^C]=[P,Q]^C.
\end{equation}
Consequently, if $P$ is Hamiltonian then $P^C$ is Hamiltonian; if $P$ and $Q$ are compatible Hamiltonian operators, then $P^C$ and $Q^C$ are compatible.
\end{theorem}

\begin{proof}
The complete lift is the tangent functor applied to the bilinear bracket.  In local jet coordinates the lifted bivector has the block form \eqref{eq:completeLift}.  The variational Schouten bracket is obtained from the Fr\'echet derivative of the coefficients and the formal adjoint.  Differentiating the coordinate expression for $[P,Q]$ in the direction $v$ gives precisely the lower-right block of $[P^C,Q^C]$, while the two off-diagonal blocks reproduce $[P,Q]$ and the upper-left block vanishes.  Hence \eqref{eq:schoutenIdentity}.  This is the formal loop-space counterpart of the classical complete-lift identity for Schouten--Nijenhuis brackets on $TM$; compare \cite{GrabowskiUrbanski1997,MitricVaisman2003}.  If $[P,P]=0$, then $[P^C,P^C]=[P,P]^C=0$.  The compatibility statement follows from bilinearity.
\end{proof}

\begin{proposition}[Hamiltonian vector fields lift tangentially]
\label{prop:HamLift}
Let $X_H=P\,\delta H/\delta u$.  Then the Hamiltonian vector field of $H^C$ with respect to $P^C$ is
\begin{equation}
\label{eq:tangentXH}
 X_{H^C}^{P^C}(u,v)=\bigl(X_H(u),\,\ell_{X_H}(v)\bigr),
\end{equation}
where $\ell_{X_H}$ denotes the Fr\'echet derivative of the base vector field.
\end{proposition}

\begin{proof}
From \eqref{eq:functionalLift},
\[
 \frac{\delta H^C}{\delta v}=\frac{\delta H}{\delta u},
\]
while $\delta H^C/\delta u$ is the formal adjoint of the linearization of $\delta H/\delta u$ applied to $v$.  Multiplication by the block operator \eqref{eq:completeLift} gives the base Hamiltonian field in the first component.  Differentiating $P\,\delta H/\delta u$ in the direction $v$ gives $P'[v]\,\delta H/\delta u$ plus the contribution from the Hessian of $H$, which is exactly the second component.
\end{proof}

\section{The tangent KdV Poisson pencil}

For the KdV equation in the normalization
\begin{equation}
\label{eq:kdv}
 u_t=u_{xxx}+6uu_x,
\end{equation}
consider the classical Hamiltonian pair \eqref{eq:scalarJ}.  We recall two consecutive Hamiltonians,
\begin{equation}
\label{eq:H01}
 H_0=\int \frac12u^2\,dx,
 \qquad
 H_1=\int\left(u^3-\frac12u_x^2\right)dx.
\end{equation}
Then
\begin{equation}
\label{eq:MagriScalar}
 u_t=J_1\frac{\delta H_0}{\delta u}
     =J_0\frac{\delta H_1}{\delta u}.
\end{equation}

The complete lifts are
\begin{equation}
\label{eq:Hhat01}
 \widehat H_0=H_0^C=\int uv\,dx,
\qquad
 \widehat H_1=H_1^C=\int\left(3u^2v-u_xv_x\right)dx.
\end{equation}
The lifted operators are exactly \eqref{eq:liftedJ}.

\begin{theorem}[Bi-Hamiltonian formulation of tangent KdV]
\label{thm:biHam}
The operators $\widehat J_0$ and $\widehat J_1$ in \eqref{eq:liftedJ} are compatible Hamiltonian operators, and \eqref{eq:tkdv} satisfies
\begin{equation}
\label{eq:biHamSystem}
 \binom{u}{v}_t
 =\widehat J_1
 \binom{\delta\widehat H_0/\delta u}{\delta\widehat H_0/\delta v}
 =\widehat J_0
 \binom{\delta\widehat H_1/\delta u}{\delta\widehat H_1/\delta v}.
\end{equation}
\end{theorem}

\begin{proof}
Compatibility follows from \cref{thm:schoutenLift} applied to the Magri pair $(J_0,J_1)$.  The variational gradients are
\[
 \grad\widehat H_0=\binom{v}{u},
 \qquad
 \grad\widehat H_1=
 \binom{6uv+v_{xx}}{3u^2+u_{xx}}.
\]
Therefore
\[
 \widehat J_0\grad\widehat H_1
 =\binom{D_x(3u^2+u_{xx})}{D_x(6uv+v_{xx})}
 =\binom{u_{xxx}+6uu_x}{v_{xxx}+6(uv)_x}.
\]
For the second representation,
\[
 J_1(u)=u_{xxx}+6uu_x,
\]
and
\begin{align*}
 J_1(v)+(4vD_x+2v_x)u
 &=v_{xxx}+4uv_x+2u_xv+4vu_x+2uv_x\\
 &=v_{xxx}+6uv_x+6u_xv.
\end{align*}
This is the second component of \eqref{eq:tkdv}.
\end{proof}

\begin{remark}
The lower-right entry $4vD_x+2v_x$ is not an ad hoc coupling.  It is exactly $J_1'[v]$.  In particular it is formally skew-adjoint, and the Jacobi identity is inherited from the base pencil rather than checked by a separate large calculation.
\end{remark}

The first lifted bracket has two elementary Casimirs.

\begin{proposition}[Casimirs and symplectic leaves]
\label{prop:casimirs}
For periodic fields, the functionals
\begin{equation}
 C_u=\int u\,dx,\qquad C_v=\int v\,dx
\end{equation}
are Casimirs of $\widehat J_0$.  On a leaf with fixed means, $D_x$ is invertible on the zero-mean tangent subspace and
\begin{equation}
\label{eq:J0inverse}
 \widehat J_0^{-1}=
 \begin{pmatrix}
 0&D_x^{-1}\\ D_x^{-1}&0
 \end{pmatrix}.
\end{equation}
Thus the first lifted bracket becomes symplectic after fixing the two mean values.
\end{proposition}

\begin{proof}
The gradients of $C_u$ and $C_v$ are the constant vectors $(1,0)^T$ and $(0,1)^T$, which lie in the kernel of $\widehat J_0$.  On the zero-mean subspace, $D_x^{-1}$ is well defined and direct block multiplication gives \eqref{eq:J0inverse}.
\end{proof}

\section{Recursion operator and the tangent hierarchy}

On a symplectic leaf of $\widehat J_0$, define the recursion operator
\[
 \widehat R=\widehat J_1\widehat J_0^{-1}.
\]
Block multiplication gives the triangular formula already announced in \eqref{eq:RhatIntro}.

\begin{theorem}[Triangular tangent recursion]
\label{thm:recursion}
Let
\begin{equation}
 R=J_1J_0^{-1}=D_x^2+4u+2u_xD_x^{-1}.
\end{equation}
Then
\begin{equation}
\label{eq:Rhat}
 \widehat R=
 \begin{pmatrix}
 R&0\\
 R'[v]&R
 \end{pmatrix},
 \qquad
 R'[v]=4v+2v_xD_x^{-1}.
\end{equation}
If $G_n[u]$ is a KdV hierarchy vector field generated by $G_{n+1}=RG_n$, then
\begin{equation}
\label{eq:Ghat}
 \widehat G_n(u,v)=\binom{G_n[u]}{\ell_{G_n}(v)}
\end{equation}
satisfies
\begin{equation}
 \widehat G_{n+1}=\widehat R\widehat G_n.
\end{equation}
Moreover, commuting base fields lift to commuting tangent fields.
\end{theorem}

\begin{proof}
The block formula follows immediately from \eqref{eq:liftedJ} and \eqref{eq:J0inverse}.  Differentiating the identity $G_{n+1}=RG_n$ in the direction $v$ gives
\[
 \ell_{G_{n+1}}(v)=R'[v]G_n+R\ell_{G_n}(v),
\]
which is precisely the second row of \eqref{eq:Rhat}.  The commutator statement is the tangent-functor identity
\[
 [\widehat G,\widehat H]=\widehat{[G,H]}.
\]
\end{proof}

Taking $G_0=u_x$ yields
\[
 \widehat G_0=(u_x,v_x)^T,
\]
and one application of $\widehat R$ gives the tangent KdV flow \eqref{eq:tkdv}.  The next member is the tangent lift of the fifth-order KdV field,
\begin{align}
 G_2[u]={}&u_{xxxxx}+10uu_{xxx}+20u_xu_{xx}+30u^2u_x,
\label{eq:KdV5}
\end{align}
with tangent component
\begin{align}
 \ell_{G_2}(v)={}&v_{xxxxx}
 +10(vu_{xxx}+uv_{xxx})
 +20(v_xu_{xx}+u_xv_{xx})\notag\\
 &+60uv\,u_x+30u^2v_x.
\label{eq:KdV5tangent}
\end{align}
Thus the coupled hierarchy is generated without solving new compatibility equations at each order.

\begin{corollary}[Heredity on the tangent hierarchy]
\label{cor:heredity}
Where the scalar KdV recursion operator is hereditary, the triangular lift \eqref{eq:Rhat} is hereditary on the tangent hierarchy generated by complete lifts.
\end{corollary}

\begin{proof}
The Nijenhuis torsion is natural under the tangent lift.  Equivalently, on the hierarchy one uses the commutativity from \cref{thm:recursion} and the standard hereditary property of $R$; compare \cite{FuchssteinerFokas1981,Dorfman1993}.
\end{proof}

\section{Lifted conservation laws and a tangent Lax representation}

The complete lift also preserves the two standard algebraic signatures of integrability: involution of Hamiltonians and zero-curvature/Lax compatibility.  This section records both facts explicitly.

\begin{proposition}[Involution is preserved by complete lift]
\label{prop:involution}
Let $H$ and $K$ be local functionals on the scalar KdV phase space and let $P$ be a Hamiltonian operator.  Then
\begin{equation}
\label{eq:liftedBracketFunctionals}
 \{H^C,K^C\}_{P^C}=\{H,K\}_P^C.
\end{equation}
In particular, every involutive Lenard--Magri family for KdV lifts to an involutive family for the tangent system.
\end{proposition}

\begin{proof}
The Hamiltonian vector field of $K^C$ is the tangent lift of $X_K=P\delta K/\delta u$ by \cref{prop:HamLift}.  Acting with this tangent vector field on $H^C=dH[v]$ differentiates the scalar identity $X_K(H)=\{H,K\}_P$ in the tangent direction.  This is exactly \eqref{eq:liftedBracketFunctionals}.
\end{proof}

Thus if $(H_n)$ is a Lenard--Magri chain satisfying
\begin{equation}
 J_1\frac{\delta H_n}{\delta u}
 =J_0\frac{\delta H_{n+1}}{\delta u},
\end{equation}
then
\begin{equation}
 \widehat J_1\grad H_n^C
 =\widehat J_0\grad H_{n+1}^C,
\end{equation}
and the lifted functionals are pairwise in involution with respect to both lifted brackets.  The first nontrivial members are $\int v\,dx$, $\widehat H_0$ and $\widehat H_1$.  Hence the tangent system carries the conservation hierarchy functorially, not by an independent search for densities.

We next lift the classical KdV Lax pair.  In our sign convention define
\begin{equation}
\label{eq:LaxBase}
 L=D_x^2+u,
 \qquad
 A=4D_x^3+6uD_x+3u_x.
\end{equation}
A direct operator computation gives
\begin{equation}
 L_t=[A,L]
 \quad\Longleftrightarrow\quad
 u_t=u_{xxx}+6uu_x.
\end{equation}
The variations in the direction $v$ are
\begin{equation}
 L'[v]=v,
 \qquad
 A'[v]=6vD_x+3v_x.
\end{equation}

\begin{theorem}[Tangent Lax pair]
\label{thm:tangentLax}
Introduce lower-triangular block differential operators
\begin{equation}
\label{eq:LaxLift}
 \widehat L=
 \begin{pmatrix}L&0\\ L'[v]&L\end{pmatrix},
 \qquad
 \widehat A=
 \begin{pmatrix}A&0\\ A'[v]&A\end{pmatrix}.
\end{equation}
Then
\begin{equation}
\label{eq:LaxLiftCompat}
 \widehat L_t=[\widehat A,\widehat L]
\end{equation}
is equivalent to the coupled tangent KdV system \eqref{eq:tkdv}.  Equivalently, if $\psi$ is a KdV eigenfunction and $\phi$ its tangent variation, the system
\begin{align}
 L\psi&=\lambda\psi,
 &L\phi+v\psi&=\lambda\phi,
\label{eq:spectralTangent}\\
 \psi_t&=A\psi,
 &\phi_t&=A\phi+A'[v]\psi
\label{eq:timeTangent}
\end{align}
is compatible if and only if $(u,v)$ satisfies \eqref{eq:tkdv}.
\end{theorem}

\begin{proof}
The upper-left block of \eqref{eq:LaxLiftCompat} is the scalar Lax equation.  Its lower-left block is
\begin{equation}
 (L'[v])_t=[A'[v],L]+[A,L'[v]],
\end{equation}
which is precisely the Fr\'echet derivative in the direction $v$ of $L_t=[A,L]$.  Differentiating the scalar KdV equation therefore gives
\[
 v_t=v_{xxx}+6u v_x+6u_xv=v_{xxx}+6(uv)_x.
\]
The eigenfunction form follows by differentiating $L\psi=\lambda\psi$ and $\psi_t=A\psi$ along the tangent direction.
\end{proof}

\begin{remark}
The block Lax pair and the block Poisson pencil have the same triangular/tangent logic.  The former lifts the spectral representation, while the latter lifts the variational bivector.  This parallel is one reason the system is a natural bridge between the geometry of differential coverings and classical integrability.
\end{remark}

\section{Finite and generalized Lie symmetries}

The tangent formulation makes the symmetry content transparent.  We first record a finite-dimensional point-symmetry algebra that will be used for reduction and for the action on elliptic moduli.

\begin{proposition}[A five-dimensional point-symmetry algebra]
\label{prop:pointSym}
System \eqref{eq:tkdv} is invariant under the vector fields
\begin{align}
 X_1&=\partial_x, &
 X_2&=\partial_t,\notag\\
 X_3&=6t\partial_x+\partial_u, &
 X_4&=x\partial_x+3t\partial_t-2u\partial_u-2v\partial_v,\notag\\
 X_5&=v\partial_v.
\label{eq:Xgenerators}
\end{align}
The nonzero brackets, up to skew-symmetry, are
\begin{equation}
\label{eq:commutators}
 [X_2,X_3]=6X_1,
 \qquad
 [X_4,X_1]=-X_1,
 \qquad
 [X_4,X_2]=-3X_2,
 \qquad
 [X_4,X_3]=2X_3,
\end{equation}
while $X_5$ is central in this subalgebra.
\end{proposition}

\begin{proof}
Translations are immediate.  The finite Galilean transformation is
\begin{equation}
\label{eq:galileanFields}
 \widetilde u(x,t)=u(x+6bt,t)+b,
 \qquad
 \widetilde v(x,t)=v(x+6bt,t),
\end{equation}
whose infinitesimal generator is $X_3$.  The scaling transformation may be written in active form as
\begin{equation}
\label{eq:scalingFields}
 u_\lambda(x,t)=\lambda^2u(\lambda x,\lambda^3t),
 \qquad
 v_\lambda(x,t)=\lambda^2v(\lambda x,\lambda^3t),
\end{equation}
and yields $X_4$ up to the conventional sign associated with active versus passive actions.  Since the $v$-equation is homogeneous in $v$, fiber scaling gives $X_5$.  Direct commutation gives \eqref{eq:commutators}.
\end{proof}

The point algebra is only the finite-dimensional shadow of the integrable symmetry structure.

\begin{theorem}[Tangent lift of generalized symmetries]
\label{thm:genSym}
Let $G[u]$ be any generalized symmetry characteristic of KdV.  Then
\begin{equation}
 \widehat G=(G,\ell_G(v))^T
\end{equation}
is a generalized symmetry of \eqref{eq:tkdv}.  If $G$ and $H$ commute, then $\widehat G$ and $\widehat H$ commute.
\end{theorem}

\begin{proof}
The second equation in \eqref{eq:tkdv} is the tangent equation of the first.  The flow generated by $G$ acts on the base solution manifold, and its differential acts on tangent vectors.  Naturality of the tangent map gives the lifted symmetry and the commutator identity.  In the language of differential coverings this is the defining role of the tangent covering; see \cite{KrasilshchikVerbovetsky2011}.
\end{proof}

\section{Traveling-wave reduction as a tangent Hamiltonian system}

Let
\begin{equation}
 z=x+ct,
 \qquad
 u(x,t)=U(z),
 \qquad
 v(x,t)=V(z).
\end{equation}
This is the invariant reduction associated with the generator $\partial_t-c\partial_x$.  Integrating once gives \eqref{eq:travelingIntro}.  The base equation admits the first integral
\begin{equation}
\label{eq:energy}
 \frac12(U')^2+U^3-\frac c2U^2-AU=E.
\end{equation}
Set $W=U'$ and define
\begin{equation}
\label{eq:hbase}
 h_{c,A}(U,W)=\frac12W^2+U^3-\frac c2U^2-AU.
\end{equation}
Then
\begin{equation}
\label{eq:baseHamODE}
 U'=\frac{\partial h}{\partial W}=W,
 \qquad
 W'=-\frac{\partial h}{\partial U}=cU-3U^2+A.
\end{equation}

The full reduced system has a particularly simple symplectic geometry.

\begin{theorem}[Hamiltonian tangent reduction]
\label{thm:reducedHamiltonian}
In coordinates $(U,W,V,Y)$ with $Y=V'$, define
\begin{equation}
\label{eq:Jred}
 \mathbb J=
 \begin{pmatrix}
 0&0&0&1\\
 0&0&-1&0\\
 0&1&0&0\\
 -1&0&0&0
 \end{pmatrix}
\end{equation}
and
\begin{equation}
\label{eq:Kred}
 \mathcal K_B(U,W,V,Y)
 =WY+(3U^2-cU-A)V-BU.
\end{equation}
Then the traveling-wave system is Hamiltonian:
\begin{equation}
\label{eq:4Dham}
 \frac{d}{dz}
 \begin{pmatrix}U\\W\\V\\Y\end{pmatrix}
 =\mathbb J\nabla\mathcal K_B.
\end{equation}
For $B=0$, $\mathcal K_0$ is the complete tangent lift of the scalar Hamiltonian $h_{c,A}$.  The term $-BU$ is an affine tangent unfolding that leaves the base orbit unchanged.
\end{theorem}

\begin{proof}
We have
\[
 \nabla\mathcal K_B=
 \begin{pmatrix}
 (6U-c)V-B\\ Y\\ 3U^2-cU-A\\ W
 \end{pmatrix}.
\]
Multiplication by \eqref{eq:Jred} gives
\[
 U'=W,\qquad
 W'=cU-3U^2+A,\qquad
 V'=Y,\qquad
 Y'=(c-6U)V+B,
\]
which is equivalent to \eqref{eq:travelingIntro}.  The complete derivative of $h_{c,A}$ along $(V,Y)$ is
\[
 h^C=WY+(3U^2-cU-A)V,
\]
so $\mathcal K_0=h^C$.
\end{proof}

Because $\mathbb J^2=-I$, the reduced system is symplectic.  Its symplectic form is represented by $\mathbb J^{-1}=-\mathbb J$; equivalently, up to the global sign convention,
\begin{equation}
 \Omega_C=dU\wedge dY+dV\wedge dW.
\end{equation}
The cross-pairing of base and tangent coordinates is characteristic of a complete tangent lift.

\section{Elliptic traveling waves and exact tangent-moduli dynamics}

Equation \eqref{eq:energy} can be written as
\begin{equation}
\label{eq:cubicCurve}
 (U')^2=-2U^3+cU^2+2AU+2E.
\end{equation}
For a cubic with distinct roots, the associated affine curve is nonsingular of genus one.  The standard shift and scaling gives an especially useful Weierstrass representation.

\begin{proposition}[Weierstrass normal form]
\label{prop:wpform}
Let
\begin{equation}
\label{eq:g2g3}
 g_2=\frac{c^2}{12}+A,
 \qquad
 g_3=-\frac E2-\frac{Ac}{12}-\frac{c^3}{216}.
\end{equation}
If
\begin{equation}
\label{eq:Delta}
 \Delta=g_2^3-27g_3^2\neq0,
\end{equation}
then every nonconstant complex traveling-wave solution is locally represented as
\begin{equation}
\label{eq:Uwp}
 U(z)=\frac c6-2\wp(z-z_0;g_2,g_3),
\end{equation}
where $\wp'^2=4\wp^3-g_2\wp-g_3$.
\end{proposition}

\begin{proof}
Put $U=c/6-2w$.  Using $w''=6w^2-g_2/2$, the profile equation $U''=cU-3U^2+A$ reduces to the identity $g_2=c^2/12+A$.  Substitution into the energy integral then gives the stated $g_3$.  Nonsingularity is exactly \eqref{eq:Delta}.  Classical details on the Weierstrass normal form may be found in \cite{WhittakerWatson,Lawden1989}.
\end{proof}

The tangent equation
\begin{equation}
\label{eq:tangentProfile}
 \cL V:=V''+(6U-c)V=B
\end{equation}
has a geometric solution space controlled by derivatives of the base elliptic family.

\begin{theorem}[Exact moduli-tangent decomposition]
\label{thm:moduliTangent}
Let $U(z;A,E,c)$ be a nonconstant local family solving the profile equation and energy relation, with the phase fixed smoothly.  Then
\begin{equation}
\label{eq:LALE}
 \cL U_A=1,
 \qquad
 \cL U_E=0,
 \qquad
 \cL U_z=0.
\end{equation}
Moreover,
\begin{equation}
\label{eq:Wronskian1}
 W(U_z,U_E):=U_zU_{zE}-U_{zz}U_E=1.
\end{equation}
Hence $U_z$ and $U_E$ form a fundamental system for the homogeneous tangent equation, and the general solution of \eqref{eq:tangentProfile} is
\begin{equation}
\label{eq:generalV}
 \boxed{V=B U_A+\alpha U_z+\beta U_E,\qquad \alpha,\beta\in\mathbb C.}
\end{equation}
\end{theorem}

\begin{proof}
Differentiate
\[
 U''+3U^2-cU-A=0
\]
with respect to $A$, $E$, and $z$.  Since $E$ does not appear in the second-order profile equation, the first three identities in \eqref{eq:LALE} follow immediately.  To prove \eqref{eq:Wronskian1}, differentiate the energy identity \eqref{eq:energy} with respect to $E$:
\[
 1=U_zU_{zE}+(3U^2-cU-A)U_E.
\]
But $U_{zz}=cU-3U^2+A$, hence
\[
 1=U_zU_{zE}-U_{zz}U_E=W(U_z,U_E).
\]
Thus the homogeneous pair is linearly independent everywhere by constancy of the Wronskian.  Since $BU_A$ is a particular solution of $\cL V=B$, formula \eqref{eq:generalV} is general.
\end{proof}

\begin{remark}[Derivatives of Weierstrass invariants]
On the elliptic locus, \eqref{eq:Uwp} gives a completely explicit moduli interpretation.  At fixed $c$ and phase,
\begin{equation}
\label{eq:wpDerivs}
 U_E=\wp_{g_3},
 \qquad
 U_A=-2\wp_{g_2}+\frac c6\wp_{g_3},
 \qquad
 U_z=-2\wp'.
\end{equation}
Here the invariant derivatives are evaluated at $(z-z_0;g_2,g_3)$.  They satisfy differential identities obtained by differentiating $\wp'^2=4\wp^3-g_2\wp-g_3$ with respect to $g_2$ and $g_3$.
\end{remark}

The identity \eqref{eq:Wronskian1} is stronger than a formal statement of linear independence: it fixes the symplectic normalization of the tangent basis.  In particular, the tangent direction associated with energy is canonically conjugate to the translation direction along the orbit.

\section{Periodic monodromy and a sharp periodicity criterion}

Assume now that $U$ belongs to a smooth family of nonconstant real periodic orbits and let
\[
 T=T(A,E,c)
\]
denote the fundamental period.  The moduli derivatives are generally not periodic because changing a parameter changes the period.  Their defect is nevertheless completely controlled.

\begin{theorem}[Monodromy of the moduli basis]
\label{thm:monodromy}
For every periodic base profile,
\begin{align}
 U_A(z+T)&=U_A(z)-T_AU_z(z),
\label{eq:UAmon}\\
 U_E(z+T)&=U_E(z)-T_EU_z(z),
\label{eq:UEmon}
\end{align}
while $U_z(z+T)=U_z(z)$.  In the homogeneous basis $(U_z,U_E)$ the monodromy has the unipotent form
\begin{equation}
\label{eq:monodromyMatrix}
 M=\begin{pmatrix}1&-T_E\\0&1\end{pmatrix}
\end{equation}
up to the ordering convention for basis vectors.  Thus the Floquet multiplier $1$ has algebraic multiplicity two, with a nontrivial Jordan block precisely when $T_E\neq0$.
\end{theorem}

\begin{proof}
The periodicity identity
\[
 U(z+T(A,E,c);A,E,c)=U(z;A,E,c)
\]
is differentiated with respect to $A$ and $E$.  Since $U_z$ is periodic, this gives \eqref{eq:UAmon}--\eqref{eq:UEmon}.  The matrix form follows immediately.
\end{proof}

\begin{corollary}[Classification of periodic tangent traveling waves]
\label{cor:periodicV}
Let $V$ be given by \eqref{eq:generalV}.  Then
\begin{equation}
\label{eq:Vjump}
 V(z+T)-V(z)=-(BT_A+\beta T_E)U_z(z).
\end{equation}
Therefore a nonconstant tangent profile is periodic with the same period as $U$ if and only if
\begin{equation}
\label{eq:periodicCondition}
 \boxed{BT_A+\beta T_E=0.}
\end{equation}
If $T_E\neq0$, there is for every $B$ a unique value
\begin{equation}
 \beta=-B\frac{T_A}{T_E}
\end{equation}
modulo the arbitrary phase mode $\alpha U_z$ that cancels the secular monodromy.
\end{corollary}

This statement is useful both analytically and geometrically.  It says that the affine forcing $B$ does not by itself destroy periodic tangent waves; rather, its induced period drift must be balanced by an energy-tangent component.

The appearance of period derivatives in the generalized kernel and monodromy of linearized KdV-type waves is closely related to the stability framework of Bronski and Johnson \cite{BronskiJohnson2010}.  The point here is different but complementary: the tangent-covering field itself is classified by the moduli derivatives, including the affine constant $B$ created by the first integration, and the cancellation condition \eqref{eq:periodicCondition} is obtained directly inside the reduced tangent Hamiltonian system rather than through a spectral Evans-function expansion.

The same result may be written in action variables.  Let
\begin{equation}
\label{eq:action}
 I(A,E,c)=\oint W\,dU,
\end{equation}
where the integral is over a closed energy orbit of \eqref{eq:baseHamODE}.  Standard one-degree-of-freedom Hamiltonian theory gives
\begin{equation}
\label{eq:IEperiod}
 I_E=T.
\end{equation}
Hence the periodicity condition becomes
\begin{equation}
\label{eq:actionCriterion}
 B I_{EA}+\beta I_{EE}=0.
\end{equation}
The tangent-wave monodromy is therefore encoded by the Hessian of the reduced action.

\section{Lie-symmetry action on the elliptic moduli}

We now connect the finite Lie symmetries from \cref{prop:pointSym} with the Weierstrass data of the reduced system.  This calculation gives a simple orbit invariant that is invisible if one studies the exact profile formula without tracking the symmetry action.

\subsection{Galilean boosts}

Under the finite boost \eqref{eq:galileanFields}, a traveling wave transforms as
\[
 \widetilde U=U+b,
 \qquad
 \widetilde c=c+6b.
\]
The integration constants must change as well.

\begin{theorem}[Galilean invariance of the Weierstrass curve]
\label{thm:GalileanModuli}
The Galilean action induced on the reduced parameters is
\begin{align}
 c_b&=c+6b,
\label{eq:cb}\\
 A_b&=A-cb-3b^2,
\label{eq:Ab}\\
 E_b&=E-Ab+\frac c2b^2+b^3.
\label{eq:Eb}
\end{align}
Under this action,
\begin{equation}
\label{eq:ginvariantGal}
 g_2(c_b,A_b)=g_2(c,A),
 \qquad
 g_3(c_b,A_b,E_b)=g_3(c,A,E).
\end{equation}
Hence the entire Weierstrass curve, including its discriminant and modular invariant, is unchanged by Galilean boosts.
\end{theorem}

\begin{proof}
Substitute $U+b$ and $c+6b$ into the profile equation.  The coefficient of $U$ cancels, leaving \eqref{eq:Ab}.  Substitution into the energy integral gives \eqref{eq:Eb}.  Finally,
\begin{align*}
 \frac{c_b^2}{12}+A_b
 &=\frac{(c+6b)^2}{12}+A-cb-3b^2
 =\frac{c^2}{12}+A,
\end{align*}
and direct substitution into $g_3$ yields the second identity.
\end{proof}

\subsection{Scaling and modular invariance}

The active scaling \eqref{eq:scalingFields} maps traveling waves to
\[
 U_\lambda(z)=\lambda^2U(\lambda z),
\]
with parameter weights
\begin{equation}
\label{eq:scalingParameters}
 c_\lambda=\lambda^2c,
 \qquad
 A_\lambda=\lambda^4A,
 \qquad
 E_\lambda=\lambda^6E.
\end{equation}
Consequently,
\begin{equation}
\label{eq:scalingg}
 g_{2,\lambda}=\lambda^4g_2,
 \qquad
 g_{3,\lambda}=\lambda^6g_3,
 \qquad
 \Delta_\lambda=\lambda^{12}\Delta.
\end{equation}

\begin{corollary}[The elliptic $j$-invariant is a Lie invariant]
\label{cor:jInvariant}
On $\Delta\neq0$, define
\begin{equation}
\label{eq:j}
 j=1728\frac{g_2^3}{\Delta}.
\end{equation}
Then $j$ is invariant under both the Galilean subgroup generated by $X_3$ and the scaling subgroup generated by $X_4$.  Over $\mathbb C$, $j$ therefore descends to the orbit space of nonsingular traveling-wave elliptic curves modulo these elementary KdV symmetries.
\end{corollary}

\begin{proof}
Galilean invariance follows from \cref{thm:GalileanModuli}.  Under scaling, numerator and denominator in \eqref{eq:j} both acquire weight $12$.  The final statement is the classical classification of complex elliptic curves by the $j$-invariant, away from the singular discriminant locus.
\end{proof}

\begin{theorem}[Complex orbit classification by elliptic moduli]
\label{thm:orbitClassification}
Let
\[
 \mathcal P=\{(c,A,E)\in\mathbb C^3:\Delta(c,A,E)\neq0\}
\]
and let $G$ be the complexified group generated by Galilean boosts $b\in\mathbb C$ and nonzero scalings $\lambda\in\mathbb C^*$.  Associate to $p=(c,A,E)$ the elliptic curve
\[
 \mathcal E_p:\quad y^2=4x^3-g_2(p)x-g_3(p).
\]
Then two points $p,q\in\mathcal P$ lie in the same $G$-orbit if and only if $\mathcal E_p$ and $\mathcal E_q$ are isomorphic over $\mathbb C$.  Equivalently,
\begin{equation}
 p\sim_G q\quad\Longleftrightarrow\quad j(p)=j(q).
\label{eq:completej}
\end{equation}
Thus the coarse complex orbit space of nonsingular traveling-wave parameters is the elliptic $j$-line.
\end{theorem}

\begin{proof}
If $p$ and $q$ are in the same orbit, \cref{thm:GalileanModuli,cor:jInvariant} imply $j(p)=j(q)$, hence the curves are isomorphic.  Conversely, apply the Galilean boosts $b=-c/6$ and $\tilde b=-\tilde c/6$ to put both triples in the gauges $c=0$ and $\tilde c=0$.  The Weierstrass pairs $(g_2,g_3)$ are unchanged by this step.  Equality of $j$ for two nonsingular short Weierstrass curves over $\mathbb C$ is equivalent to the existence of $\lambda\in\mathbb C^*$ such that
\[
 \tilde g_2=\lambda^4g_2,
 \qquad
 \tilde g_3=\lambda^6g_3.
\]
Apply the scaling action with this $\lambda$.  In the gauge $c=0$, the reduced parameters are recovered uniquely from the invariants by
\[
 A=g_2,
 \qquad
 E=-2g_3.
\]
Hence the scaled first triple coincides with the second gauged triple, and reversing the Galilean gauges proves that $p$ and $q$ lie in the same $G$-orbit.
\end{proof}

\begin{remark}[Weighted-projective interpretation]
After the Galilean gauge $c=0$, the remaining scaling has weights $(4,6)$ on $(g_2,g_3)$.  Accordingly the quotient before passing to the coarse invariant is the familiar weighted-projective Weierstrass moduli picture; the function $j$ is its coarse coordinate on the nonsingular locus.  This gives a concrete Lie-symmetry realization of elliptic moduli directly from the KdV traveling-wave parameter space.
\end{remark}

\begin{remark}
For real traveling waves, $j$ does not by itself encode every real-form distinction or the placement of a particular real oval.  The statement in \cref{cor:jInvariant} is therefore naturally complex-algebraic.  The real dynamics additionally depends on the ordering of real roots and on the selected connected component of the real curve.
\end{remark}

\section{A verified periodic example}

To illustrate the monodromy theorem without replacing the structural analysis by numerics, take
\begin{equation}
\label{eq:exampleParams}
 c=4,\qquad A=-1,\qquad E=0.02.
\end{equation}
Then
\begin{equation}
 g_2=\frac13,
 \qquad
 g_3=0.027037037\ldots,
 \qquad
 \Delta=0.0173>0,
 \qquad
 j=3699.421965\ldots.
\end{equation}
The cubic in \eqref{eq:cubicCurve} has three real roots
\begin{equation}
 0.0208613098\ldots,
 \quad 0.8462692611\ldots,
 \quad 1.1328694292\ldots,
\end{equation}
and the oval between the two upper roots is periodic.  Starting at $U(0)=1$, $U_z(0)=0.2$, high-accuracy integration gives
\begin{equation}
\label{eq:numericalT}
 T=4.532817921155\ldots.
\end{equation}
For the canonical initial normalization $U_E(0)=0$ and $W(U_z,U_E)=1$, and for $U_A(0)=0$ with fixed energy at the initial point, the computed derivatives are
\begin{equation}
 T_E=4.862742472047\ldots,
 \qquad
 T_A=1.025482208697\ldots.
\end{equation}
The numerical Wronskian after one full period is
\begin{equation}
 W(U_z,U_E)(T)=1.00000000000000
\end{equation}
to the reported precision, independently confirming \eqref{eq:Wronskian1}.  For $B=1$, the periodicity criterion selects
\begin{equation}
 \beta=-\frac{T_A}{T_E}=-0.210885568091\ldots.
\end{equation}
The phase orbit and the cancellation of the quasi-periodic tangent drift are shown in \cref{fig:phase,fig:tangent}.

\begin{figure}[ht]
\centering
\includegraphics[width=0.62\textwidth]{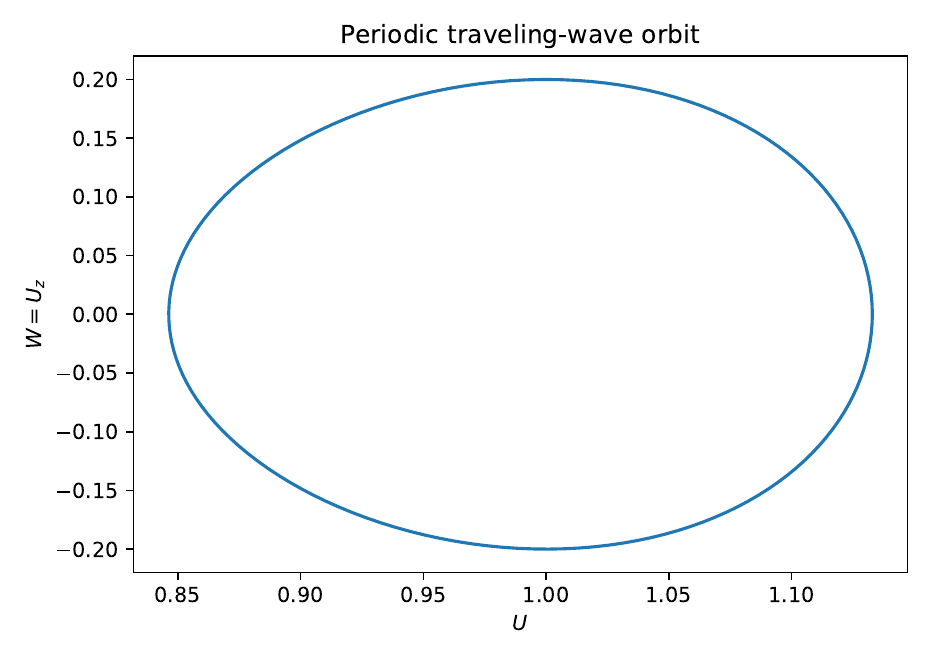}
\caption{Periodic traveling-wave orbit in the reduced $(U,W)$ phase plane for \eqref{eq:exampleParams}.}
\label{fig:phase}
\end{figure}

\begin{figure}[ht]
\centering
\includegraphics[width=0.68\textwidth]{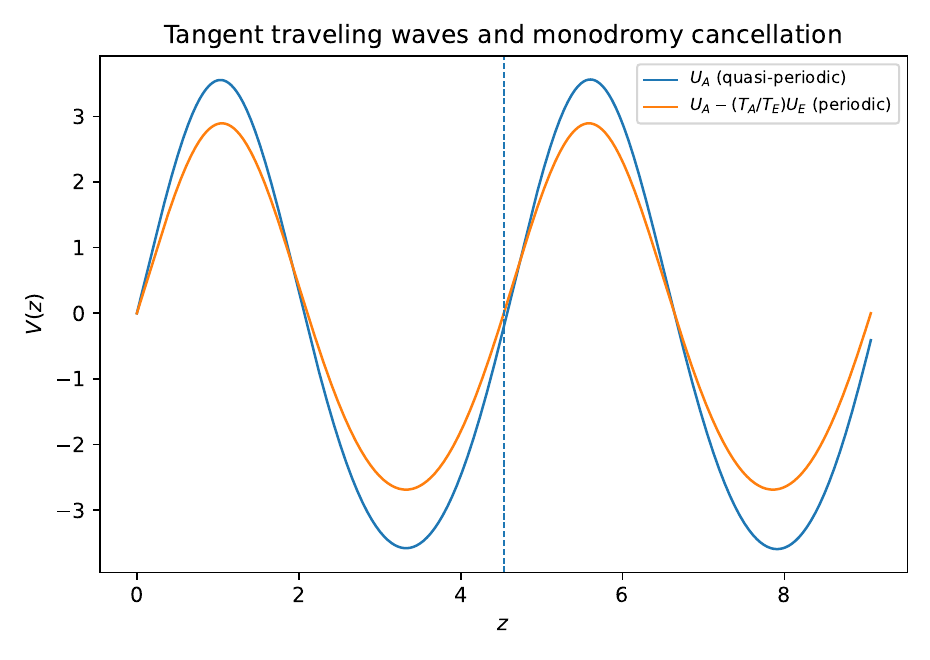}
\caption{The particular moduli tangent $U_A$ is quasi-periodic, while the combination $U_A-(T_A/T_E)U_E$ is periodic, exactly as predicted by \cref{cor:periodicV}.  The dashed line marks one base period.}
\label{fig:tangent}
\end{figure}

\section{Geometric interpretation and relation to other coupled KdV systems}

There are many coupled KdV-type equations in the literature, arising from constrained Hamiltonian systems, moving-curve invariants, matrix Lax equations, or algebraic extensions; see, for example, \cite{MariBeffa2000,RestucciaSotomayor2016}.  The geometry of \eqref{eq:tkdv} is different in a specific way: the second component is not an independent nonlinear field chosen to produce a new two-component spectral problem.  It is the tangent vector to the KdV solution manifold.  Consequently, the coupled Poisson pencil, recursion operator and hierarchy are compelled by functoriality.

This tangent origin has several consequences.

First, the lower-right block of every lifted Hamiltonian operator is its Fr\'echet variation.  The compatibility of the pencil is therefore inherited from a Schouten identity rather than established by coefficient matching.  Second, generalized symmetries lift together with their linearizations, making the triangular recursion formula unavoidable.  Third, the traveling-wave equation for $V$ is the Jacobi equation of the scalar reduced Hamiltonian system, with the single affine term $B$ encoding variation of the profile integration constant.  Finally, moduli derivatives provide an exact basis for tangent waves; in the periodic case the same derivatives directly encode the monodromy of the variational equation.

From the viewpoint of Poisson geometry, this is the infinite-dimensional analogue of passing from a Poisson manifold to its tangent Poisson manifold and then reducing a Hamiltonian flow.  The reduced symplectic pairing in \cref{thm:reducedHamiltonian} makes this analogy literal at the traveling-wave level.  From the viewpoint of Lie theory, Galilean and scaling symmetries act not only on the fields but on the algebraic parameters of the elliptic curve, and $j$ is the induced modular invariant.  These two structures intersect in the reduced moduli space rather than merely coexisting in separate sections of the calculation.

\section{Conclusions}

We have developed a geometric analysis of the tangent KdV system \eqref{eq:tkdv} centered on the interaction of Poisson, Lie and elliptic structures.  The main conclusions are as follows.

The Magri pencil lifts to the explicit matrix pencil \eqref{eq:liftedJ}, and the lifted Hamiltonians \eqref{eq:Hhat01} generate the same coupled field through two compatible brackets.  The recursion operator has the triangular complete-lift structure \eqref{eq:Rhat}, which transports the whole KdV hierarchy and its commutativity to the tangent covering.  A finite point-symmetry algebra contains translations, Galilean boosts, scaling and fiber scaling, while every generalized KdV symmetry has a canonical tangent lift.

Traveling-wave reduction produces the complete tangent Hamiltonian system \eqref{eq:4Dham}.  On the elliptic locus, the exact formula \eqref{eq:Uwp} allows the tangent equation to be solved geometrically by moduli derivatives.  The identity $W(U_z,U_E)=1$ is a canonical normalization, and the periodic monodromy formulas \eqref{eq:UAmon}--\eqref{eq:UEmon} yield the sharp criterion \eqref{eq:periodicCondition}.  Thus the secular behavior of a tangent traveling wave is controlled entirely by derivatives of the base period.

Finally, the Galilean action \eqref{eq:cb}--\eqref{eq:Eb} preserves the Weierstrass invariants, while scaling acts with weights $(4,6)$.  The modular invariant $j$ is therefore preserved by both Lie actions and descends to the nonsingular elliptic orbit space over $\mathbb C$.

Several extensions are natural.  The same construction can be applied to higher KdV flows, to finite-gap stationary reductions, and to tangent lifts of other bi-Hamiltonian hierarchies.  A second direction is to place the tangent reduction in a fully global Poisson-groupoid framework, where the complete lift and symmetry quotient can be compared with Lie algebroid integration.  A third is spectral: the monodromy formulas derived here suggest a direct relation between tangent moduli, period derivatives and the neutral Floquet modes that appear in stability theory.  These questions lie at the interface of Poisson geometry, Lie symmetry and integrable dynamics.

\section*{Data and code availability}
No external data are used.  Symbolic checks of the bi-Hamiltonian identities, the Galilean invariance of $(g_2,g_3)$, and the Weierstrass reduction, together with the numerical monodromy verification used for \cref{fig:phase,fig:tangent}, are supplied as reproducibility scripts with the manuscript.

\section*{Conflict of interest}
The author declares no conflict of interest.

\section*{Funding}
No external funding was received for this work.

\appendix
\section{Direct algebraic checks for the lifted Hamiltonians}

For completeness, we collect the short computations underlying \cref{thm:biHam}.  Since
\[
 \frac{\delta\widehat H_0}{\delta u}=v,
 \qquad
 \frac{\delta\widehat H_0}{\delta v}=u,
\]
the upper component of $\widehat J_1\grad\widehat H_0$ is
\[
 (D_x^3+4uD_x+2u_x)u=u_{xxx}+6uu_x.
\]
The lower component is
\begin{align*}
 &(D_x^3+4uD_x+2u_x)v+(4vD_x+2v_x)u\\
 &\qquad=v_{xxx}+4uv_x+2u_xv+4vu_x+2uv_x\\
 &\qquad=v_{xxx}+6(uv)_x.
\end{align*}
Similarly,
\[
 \frac{\delta\widehat H_1}{\delta u}=6uv+v_{xx},
 \qquad
 \frac{\delta\widehat H_1}{\delta v}=3u^2+u_{xx},
\]
and multiplication by $\widehat J_0$ gives \eqref{eq:tkdv}.

\section{Derivation of the Galilean parameter action}

Let $\widetilde U=U+b$ and $\widetilde c=c+6b$.  Then
\begin{align*}
 \widetilde U''+3\widetilde U^2-\widetilde c\widetilde U
 &=U''+3U^2-cU-cb-3b^2,
\end{align*}
which yields $A_b=A-cb-3b^2$.  Substitution into the energy integral gives
\[
 E_b=E-Ab+\frac c2b^2+b^3.
\]
A direct calculation then verifies
\begin{align*}
 \frac{c_b^2}{12}+A_b&=\frac{c^2}{12}+A,\\
 -\frac{E_b}{2}-\frac{A_bc_b}{12}-\frac{c_b^3}{216}
 &=-\frac E2-\frac{Ac}{12}-\frac{c^3}{216}.
\end{align*}
Thus the Galilean action changes the affine embedding of the reduced cubic but not its Weierstrass invariants.

\section{Reproducibility note}

The file \texttt{symbolic\_checks.py} verifies the following identities symbolically:
\begin{enumerate}[label=(\roman*)]
\item the two Hamiltonian representations in \cref{thm:biHam};
\item the tangent linearization of the fifth-order KdV field \eqref{eq:KdV5};
\item the Weierstrass substitution \eqref{eq:Uwp};
\item invariance of $g_2$ and $g_3$ under \eqref{eq:cb}--\eqref{eq:Eb}.
\end{enumerate}
The file \texttt{numerical\_check.py} integrates the reduced base and tangent equations for \eqref{eq:exampleParams}, computes the first return period, verifies the unit Wronskian, evaluates $T_A$ and $T_E$, and generates the two figures.  The numerical experiment is illustrative only; none of the structural theorems depends on it.


\begin{thebibliography}{99}

\bibitem{GardnerGreeneKruskalMiura1967}
C.~S. Gardner, J.~M. Greene, M.~D. Kruskal and R.~M. Miura,
Method for solving the Korteweg--de Vries equation,
\emph{Phys. Rev. Lett.} \textbf{19} (1967), 1095--1097.

\bibitem{Lax1968}
P.~D. Lax,
Integrals of nonlinear equations of evolution and solitary waves,
\emph{Comm. Pure Appl. Math.} \textbf{21} (1968), 467--490.

\bibitem{Magri1978}
F. Magri,
A simple model of the integrable Hamiltonian equation,
\emph{J. Math. Phys.} \textbf{19} (1978), 1156--1162.

\bibitem{FuchssteinerFokas1981}
B. Fuchssteiner and A.~S. Fokas,
Symplectic structures, their B\"acklund transformations and hereditary symmetries,
\emph{Physica D} \textbf{4} (1981), 47--66.

\bibitem{Dorfman1993}
I. Dorfman,
\emph{Dirac Structures and Integrability of Nonlinear Evolution Equations},
John Wiley \& Sons, 1993.

\bibitem{Olver1993}
P.~J. Olver,
\emph{Applications of Lie Groups to Differential Equations}, 2nd ed.,
Springer, New York, 1993.

\bibitem{Vaisman1994}
I. Vaisman,
\emph{Lectures on the Geometry of Poisson Manifolds},
Birkh\"auser, Basel, 1994.

\bibitem{GrabowskiUrbanski1997}
J. Grabowski and P. Urba\'nski,
Tangent lifts of Poisson and related structures,
\emph{J. Phys. A: Math. Gen.} \textbf{28} (1995), 6743--6777.

\bibitem{MitricVaisman2003}
G. Mitric and I. Vaisman,
Poisson structures on tangent bundles,
\emph{Differential Geom. Appl.} \textbf{18} (2003), 207--228.

\bibitem{KerstenKrasilshchikVerbovetsky2004}
P. Kersten, J. Krasil'shchik and A. Verbovetsky,
Hamiltonian operators and $\ell^*$-coverings,
\emph{J. Geom. Phys.} \textbf{50} (2004), 273--302.

\bibitem{KrasilshchikVerbovetsky2011}
J. Krasil'shchik, A. Verbovetsky and R. Vitolo,
A unified approach to computation of integrable structures,
\emph{Acta Appl. Math.} \textbf{112} (2010), 161--186.

\bibitem{MariBeffa2000}
G. Mar\'i Beffa,
The theory of differential invariants and KdV Hamiltonian evolutions,
\emph{Bull. Soc. Math. France} \textbf{127} (1999), 363--391.

\bibitem{RestucciaSotomayor2016}
A. Restuccia and A. Sotomayor,
Full Hamiltonian structure for a parametric coupled KdV system,
\emph{Open Phys.} \textbf{14} (2016), 95--105.

\bibitem{BronskiJohnson2010}
J.~C. Bronski and M.~A. Johnson,
The modulational instability for a generalized Korteweg--de Vries equation,
\emph{Arch. Ration. Mech. Anal.} \textbf{197} (2010), 357--400.

\bibitem{WhittakerWatson}
E.~T. Whittaker and G.~N. Watson,
\emph{A Course of Modern Analysis}, 4th ed.,
Cambridge University Press, 1927.

\bibitem{Lawden1989}
D.~F. Lawden,
\emph{Elliptic Functions and Applications},
Springer, 1989.

\bibitem{ArnoldMMCM}
V.~I. Arnold,
\emph{Mathematical Methods of Classical Mechanics}, 2nd ed.,
Springer, 1989.

\end{thebibliography}
\end{document}